\documentclass[singlecolumn,laps,pra,amsmath,amssymb]{revtex4}
\usepackage{amsfonts,amssymb,amsmath}      
\usepackage{dcolumn}
\usepackage{bm}
\usepackage{hyperref}
\usepackage{epstopdf}
\usepackage{graphicx}
\usepackage{mathtools}
\usepackage{amsthm}
\usepackage{color}
\usepackage{mathtools}
\usepackage{algorithm2e}
\allowdisplaybreaks

\newcommand{\nn}{\nonumber \\}

\newcommand{\ket}[1]{|{#1}\rangle}

\def\lsim{\mathrel{\rlap{\lower4pt\hbox{$\sim$}}
    \raise1pt\hbox{$<$}}}                
\def\gsim{\mathrel{\rlap{\lower4pt\hbox{$\sim$}}
    \raise1pt\hbox{$>$}}}                

\newcommand{\lem}[1]{\hyperref[lem:#1]{Lemma~\ref*{lem:#1}}}

\newcommand{\Hamseg}{V}
\newcommand{\appseg}{\tilde V}
\newcommand{\oaa}{V_{\rm oaa}}

\newcommand{\corr}{V_C}
\newcommand{\corrp}{V^+_C}
\newcommand{\appcor}{\tilde V_C}

\newcommand{\obv}{OAA}
\newcommand{\vd}{V_\Delta}
\newcommand{\sa}{s_A}
\newcommand{\appu}{\tilde U}

\newtheorem{theorem}{Theorem}
\newtheorem{lemma}[theorem]{Lemma}

\begin{document}
\title{Improved Hamiltonian simulation via a truncated Taylor series and corrections}
\author{Leonardo Novo${}^{1,2,3}$ and Dominic W. Berry${}^3$}
\affiliation{${}^1$ Instituto de Telecomunica\c{c}\~oes, Physics of Information and Quantum Technologies Group, Portugal}
\affiliation{${}^2$Instituto Superior T\'{e}cnico, Universidade de Lisboa, Portugal}
\affiliation{${}^3$Department of Physics and Astronomy, Macquarie University, Sydney, NSW 2109, Australia}
\date{\today}
   
  \begin{abstract}
We describe an improved version of the quantum simulation method based on the implementation of a truncated Taylor series of the evolution operator. The idea is to add an extra step to the previously known algorithm which implements an operator that corrects the weightings of the Taylor series. This way, the desired accuracy is achieved with an  improvement in the overall complexity of the algorithm. This quantum simulation method is applicable to a wide range of Hamiltonians of interest, including to quantum chemistry problems.
  \end{abstract}
     \maketitle
     
\section{Introduction}
The problem of simulating quantum mechanical evolution was one of the main motivations for the proposal of quantum computers \cite{Feyn}, because the exponential growth of the Hilbert space dimension means that it quickly becomes intractable for classical computers.
Lloyd was the first to explicitly show that a quantum computer can be used as a universal quantum simulator \cite{lloyd} for local quantum systems, by decomposing the evolution operator into a set of quantum gates. Aharonov and Ta-Shma considered the alternative scenario where the Hamiltonian is sparse and there is an efficient procedure to calculate its nonzero entries \cite{ats}. Since then, several improved quantum simulation algorithms were proposed \cite{BACS,wiebe11,childs10,PQSV11,BC12,CW12, BCCKS14, BCCKS15, BCK15}.

Recently, nearly optimal methods for Hamiltonian simulation have been developed \cite{BCCKS14,BCCKS15,BCK15} which achieve an exponential improvement of the complexity in the dependence on the simulation error, $\epsilon$. For sparse Hamiltonian simulation the lower bound on the complexity of \begin{equation}
O\left(\tau+\frac{\log(1/\epsilon)}{\log{\log(1/\epsilon)}}\right),
\end{equation}
was proven in \cite{BCK15}, with $\tau=t\|H\|_{\max} d$, where $d$ is the sparsity, $\|H\|_{\max}$ is the max-norm of the Hamiltonian and $t$ is the evolution time.
This bound is for the \emph{query} complexity, which is the number of calls to oracles for calculating the positions and values of nonzero entries in the Hamiltonian (expressed as a matrix in the computational basis).
The nearly optimal method of \cite{BCK15} has a query complexity of
\begin{equation}
O\left(\tau\frac{\log(\tau/\epsilon)}{\log{\log(\tau/\epsilon)}}\right),
\end{equation}
whereas in \cite{BCCKS14} a similar complexity is achieved, where $\tau$ is replaced by $\tau'=t\|H\|_{\max} d^2$. 
In these methods, the dependence on the time and error appears in the form of a product, in contrast to the proven lower bound, where this dependence appears as a sum. Very recently, the desired dependence as a sum was achieved for sparse Hamiltonian simulation with two different approaches: one involving quantum signal processing techniques \cite{low} and another involving a corrected quantum walk \cite{berry2016corrected}.
However, whether this improvement could be achieved for the simulation of Hamiltonians given by sums of local terms was an open question.
More recently, the signal processing approach was generalized to address the simulation of Hamiltonians given by sums of local terms \cite{low2}.

In this work, we apply an approach similar to that of the corrected quantum walk \cite{berry2016corrected} to improve the complexity of the Hamiltonian simulation method based on a truncated Taylor series \cite{BCCKS15}. The original method described in Ref.~\cite{BCCKS15} is applicable to Hamiltonians given by sums of unitary terms $H=\sum_{\ell} \alpha_{\ell} H_{\ell}$, with $\alpha_{\ell}\geq0$ and has a complexity of 
\begin{equation}
O\left(T \frac{\log{(T/\epsilon)}}{\log\log(T/\epsilon)}\right),
\end{equation} 
where $T:=\sum_{\ell} \alpha_{\ell} t$.
This complexity is in terms of controlled-$H_\ell$ operations (which is equivalent to the query complexity if the Hamiltonian is specified by an oracle).
Our approach is based on using this method to achieve a quantum simulation with a fixed error $\delta$, and then applying a correction operator in order to achieve an error less than $\epsilon$.
This way we obtain a complexity of 
\begin{equation}
O\left(T \frac{\log{(T)}}{\log\log(T)}+\log(1/\epsilon)\right).
\end{equation}
This leads to a close to quadratic improvement with respect to the original approach when the error is on the order of $\exp{(-T)}$. The truncated Taylor series approach is applicable to a wide range of problems where the Hamiltonian naturally decomposes as a sum of terms. These include, for example, the important problem of simulating quantum chemistry \cite{chem1,chem2}, which might be one of the first applications of quantum computers due to its relatively low resource requirements. Another advantage of the truncated Taylor series method is that it requires less additional gates than the algorithms based on quantum walks (that is, gates that are additional to controlled-$H_\ell$ operations or calls to an oracle for the Hamiltonian).

This paper is structured as follows: in Section \ref{sec:2} we provide a summary of the truncated Taylor series simulation algorithm and of the improved approach using a correction; in Section \ref{sec:corrected_Taylor} we show the main result of this work regarding the complexity of the corrected approach. Finally, we present the conclusion in Section \ref{conclusions}. 
 
\section{Background and Summary of the method}\label{sec:2}
\subsection{Simulating Hamiltonian dynamics with a Truncated Taylor series}\label{background}
In Ref.~\cite{BCCKS15}, a method is presented to perform a quantum simulation of the unitary operator $U=\exp(-i H t)$, for a given Hamiltonian $H$, up to accuracy $\delta$. The main advantage of this method is that the dependence of the number of gates on the precision is $\approx \log(1/\delta)$, whereas methods based on Suzuki-Trotter expansion have a polynomial dependence on $1/\delta$.

The algorithm is as follows. First, decompose the Hamiltonian as 
\begin{equation}\label{H_decomposition}
H=\sum_{\ell=1}^L \alpha_{\ell} H_{\ell}\, ,
\end{equation}    
where each $H_{\ell}$ is unitary. This decomposition is general, since any Hamiltonian can be written in such form. In fact, in many problems of interest the interactions are local and the Hamiltonian can be decomposed into a small number of easy to implement terms. Also, without loss of generality, we can define $H_{\ell}$ such that $\alpha_{\ell}>0$. Furthermore, we define $T:=\sum_{\ell=1}^L\alpha_{\ell} t$. This Hamiltonian simulation method is based on the implementation of an approximate version of the unitary 
\begin{equation}\label{hseg}
\Hamseg:=\exp(-i H t/r) \, ,
\end{equation}
where $r$ is the number of segments into which the time is divided. This operator can be approximated by the truncated Taylor series
\begin{equation}\label{utilde}
\appseg:=\sum_{k=0}^K \frac{1}{k!}(-i H t/r)^k.
\end{equation}
In this work, we take the order $K$ to be at least $2$.
If $\|\Hamseg-\appseg\|<\delta/r$, then we ensure that $\| U-\appseg^r\|<\delta$. It can be shown that, if $T/r$ is a constant, the error $\delta/r$ is bounded by $O(K^ {-K})$ and thus the truncation of the Taylor series at  
 \begin{equation}\label{scaling_K}
 K=O\left(\frac{\log(r/\delta)}{\log\log(r/\delta)}\right)
 \end{equation}
is enough to achieve the desired accuracy.
 
The implementation of $\appseg$ is achieved as follows. The unitary $\appseg$ can be expanded as a sum of unitaries as
 \begin{align}\label{expansion}
 \appseg&=\sum_{k=0}^K\sum_{\ell_1...\ell_k=1}^L \frac{(-it/r)^k}{k!}\alpha_{\ell_1}...\alpha_{\ell_k}H_{\ell_1}...H_{\ell_k}\nn
 &=\sum_{j\in \tilde{J}}\beta_j V_j \, ,
 \end{align}
where we define the truncated index set
\begin{equation}
\tilde{J}\coloneqq
 \{(k,\ell_1,...,\ell_k):k\in\{0,...,K\},\ell_1,...,\ell_k\in\{1,...L\}\} \, ,
\end{equation} 
such that $V_j$ is a unitary and $\beta_j$ is a positive coefficient, and are given by
\begin{align}
V_j := (-i)^k H_{\ell_1}\dots H_{\ell_k}\, , \qquad
\beta_j := \frac{(t/r)^k}{k!}\alpha_{\ell_1}...\alpha_{\ell_k}\, .
\end{align}
To implement this operator on a quantum computer, the following mechanism is introduced
\begin{equation}\label{selectV}
\text{select}(V)\ket{j}\ket{\psi}=\ket{j}V_j\ket{\psi} \, ,
\end{equation}
which is an operation controlled on the ancilla states $\ket{j}$, and $\ket{\psi}$ is some quantum state.
This operator can be implemented using $K$ controlled-$H_\ell$ operations that act on a state of the form $\ket{b}\ket{\ell}\ket{\psi}$ and give $\ket{b}\ket{\ell}(-iH_\ell)^b\ket{\psi}$ \cite{BCCKS15}.
Henceforth in this paper, when we refer to the complexity of the algorithm or a certain operation, it will be in terms of the number of controlled-$H_\ell$ operations used.
The cost of this operation in terms of universal gates was analyzed in Ref.~\cite{BCCKS15}. 

Furthermore, a unitary transformation $B$ acting on the ancillas as
\begin{equation}\label{eq:B}
B\ket{0}=\frac{1}{\sqrt{s}}\sum_{j\in \tilde{J}}\sqrt{\beta_j}\ket{j}\, ,
\end{equation}
is defined, where $s=\sum_{j\in \tilde{J}}\beta_j$ is the normalization factor.
It can be shown that the operator $W$ given by
\begin{equation}
W:= (B^{\dagger}\otimes\mathbb{I})\left[\text{select}(V)\right](B\otimes \mathbb{I})\, ,
\end{equation}
acts as
\begin{equation}\label{W_operator}
W\ket{0}\ket{\psi}=\frac{1}{s}\ket{0}\appseg\ket{\psi}+\sqrt{1-\frac{1}{s^2}}\ket{\phi}\, ,
\end{equation}
where the ancillary state of $\ket{\phi}$ has no support in $\ket{0}$. The component of the wavefunction that we are interested in is the one flagged by the ancilla qubit being in state $\ket{0}$. This amplitude decreases with $s$ but can be amplified to $1$ using Oblivious Amplitude Amplification (OAA) \cite{BCCKS15}.
When $s<2$, it is trivial to include an ancilla qubit to increase $s$ to $2$, as discussed in \cite{BCCKS15}.
For simplicity in the following discussion, if $s$ would otherwise be less $<2$, we will assume that the ancilla qubit is added to increase $s$ to $2$.
When $s=2$, one step of OAA is enough to amplify the amplitude of the component with $\ket{0}\appseg\ket{\psi}$ in Eq.~\eqref{W_operator} to $\approx 1$.  
It can be shown that  
\begin{equation}
s= \sum_{j\in \tilde{J}} \beta_j= e^{T/r}-\sum_{k=K+1}^{\infty} \frac{(T/r)^k}{k!}= e^{T/r}+ O(\delta/r)\, .
\end{equation}  
Hence, we would like the number of segments $r$ to scale as $T$ so that OAA only gives an $O(1)$ overhead to the algorithm.
In Ref.~\cite{BCCKS15}, $r$ is taken to be the smallest integer such that $s\le 2$, which corresponds to $r\approx T/\ln(2)$.
Here we choose $r$ to be a larger integer satisfying $r>4T$.

For $s=2$, after one step of OAA, we obtain the state 
\begin{equation}\label{onesegment}
\ket{\Psi}=\ket{0}\oaa \ket{\psi}+ \ket{\perp}\, ,
\end{equation}
where the ancilla components of the state $\ket{\perp}$ are orthogonal to $\ket{0}$, and 
\begin{equation}\label{voaa}
\oaa =\appseg\left(\frac{3}{2}\openone -\frac{1}{2}\appseg^{\dagger}\appseg\right).
\end{equation}
Using Lemma 6 of \cite{BCK15} $\left\|\Hamseg-\oaa \right\|=O(\delta/r)$, which implies that $\left\|\ket{\Psi}-\ket{0}\Hamseg\ket{\psi}\right\|=O(\delta/r)$. By repeating this process $r$ times using fresh ancillas, we obtain the state
\begin{equation}\label{finalstate}
\ket{\Psi_r}=\ket{0}^{\otimes r}\otimes \oaa^r\ket{\psi}+\ket{\perp_2} \, .
\end{equation}
This way, we achieve the desired accuracy since 
\begin{equation}
\left\|\oaa^r-U\right\|=O(\delta).
\end{equation}
The complexity of the algorithm stems from the fact that the implementation of a truncated Taylor series at order $K$ has a complexity of $O(K)$.
Since this must be done $r$ times and $r=O(T)$, the overall complexity in  Ref.~\cite{BCCKS15} is
\begin{equation}
O\left(T \frac{\log{(T/\delta)}}{\log\log(T/\delta)}\right).
\end{equation}
In the next subsection we summarize the idea of how a better scaling can be achieved by applying a correction operator at the end of this simulation procedure. 
\subsection{Summary of the corrected Taylor series approach}
In the corrected Taylor series approach we consider $\delta$ to be a constant less than $1/2$,
which is the accuracy of the first part of the Hamiltonian simulation algorithm, and we would like to achieve a final accuracy of the simulation algorithm of $\epsilon< \delta$.
To do so, we implement a correction operator $\appcor$ such that 
\begin{equation}\label{propcorrection}
\left\|\appcor \oaa^r-U\right\|<\epsilon \, .
\end{equation}
Since $\oaa^r$ is a function of $\appseg$, which is a series of powers of the Hamiltonian $H$, we can write the correction operator as 
\begin{equation}\label{defcorrection}
\appcor=\sum_{k=0}^Q a_k \mathcal{H}^k ,
\end{equation}
where $\mathcal{H}=-iH$.
The coefficients $a_k$ are chosen such that, if there was an infinite sum, the correction would give exactly the desired operation $U$.
The truncation order $Q$ is chosen so as to achieve the desired accuracy $\epsilon$.
The implementation of $\appcor$ can be achieved in an analogous way as for $\appseg$.
That is, after obtaining the state in Eq.~\eqref{finalstate}, another ancilla is appended of dimension $Q+1$,
and the procedure for implementing sums of unitaries is followed in order to obtain the corrected state
\begin{equation}\label{aftercor}
\ket{\Psi_C}=\ket{0}^{\otimes r+1}\otimes \frac{1}{s_C}\appcor \oaa^r\ket{\psi}+\ket{\perp_3}\, ,
\end{equation}
with $s_C=\sum_{k=0}^Q|a_k|A^k$, for $A=\sum_{\ell=1}^L \alpha_{\ell}$. We show in Lemma \ref{lemma1} of Sec.~\ref{sec:corrected_Taylor} that $s_C\le 2$ and so it is possible to use a single step of OAA to amplify the component of the state with $\appcor \oaa^r\ket{\psi}$.
Let us define $\appu=\appcor \oaa^r$.
After OAA, we have implemented the operator 
\begin{equation}
\appu_{\rm oaa}=\appu\left(\frac{3}{2}\openone-\frac{1}{2}\appu^{\dagger}\appu\right) ,
\end{equation}
which is analogous to Eq.~\eqref{voaa}. From Eq.~\eqref{propcorrection} and Lemma 6 of \cite{BCK15}, this operator is within $O(\epsilon)$ of the desired operator $U$, so the desired accuracy is achieved.

The other crucial factor for the complexity is the order $Q$ at which $\appcor$ is truncated. We show that the desired accuracy $\epsilon$ is achieved for 
\begin{equation}
Q=O\left(T+\log(1/\epsilon)\right)\, ,
\end{equation}
which gives the complexity of the implementation of the correction. The final algorithmic complexity after the correction is thus
\begin{equation}\label{new_bound}
O\left(T \frac{\log{(T)}}{\log\log(T)}+\log(1/\epsilon)\right).
\end{equation}
This is an improvement over the complexity from Ref.~\cite{BCCKS15} of
\begin{equation}\label{old_bound}
O\left(T \frac{\log{(T/\epsilon)}}{\log\log(T/\epsilon)}\right) ,
\end{equation}
since the dependence on the error and time appear as a sum and not as a product.
In particular, if $\epsilon\approx \exp{(-T)}$, which can happen if we need a very low error and/or the simulation time is small, the complexity from Eq.~\eqref{new_bound} gives a close to quadratic improvement over the complexity in Eq.~\eqref{old_bound}. 

When applied to the problem of sparse Hamiltonian simulation, we have $T\leq \tau'=t\|H\|_{\max} d^2$, where $d$ is the sparsity.
Hence, the bound obtained with the corrected Taylor series approach is closer to the proven lower bound from Ref.~\cite{BCK15} of 
\begin{equation}
O\left(\tau+\frac{\log(1/\epsilon)}{\log{\log(1/\epsilon)}}\right),
\end{equation}
with $\tau=t\|H\|_{\max} d$.
Although the dependence on the sparsity is in general worse, for particular applications there is often a known decomposition of the form of Eq.~\eqref{H_decomposition}, so the square in the dependence on $d$ can be eliminated.
A better dependence on sparsity can be obtained using other approaches \cite{berry2016corrected,low}.  

Before proceeding with the technical proofs in the next section, we summarize the proposed algorithm. We assume we are given a state $\ket{\psi}$ which represents the initial state of the quantum system whose dynamics we want to simulate.
\RestyleAlgo{boxruled}
\LinesNumbered
\begin{algorithm}[h!]
\caption{Hamiltonian simulation with a corrected Taylor series \label{alg1}}
\begin{enumerate}
\item For segments $1$ to $r$ with $r\in \Theta(T)$, perform the following steps.
\begin{enumerate}
\item Append an ancilla of dimension $K+1$ in state $\ket{0}$ and apply the operator $B$ as in Eq.~\eqref{eq:B}.
\item Perform the controlled unitary operation ${\rm select}(V)$ given in Eq.~\eqref{selectV}.
\item Apply the operator $B^{\dagger}$ to the ancilla to obtain the state described in Eq.~\eqref{W_operator}. 
\item Apply one step of OAA, as described in Lemma 5 of Ref.~\cite{BCK15}.
This results in the implementation of the\\ operation $\oaa$, as defined in Eq.~\eqref{voaa}, with success flagged by a zero in the ancilla.
\end{enumerate}
\item Apply the correction $\appcor$ defined in Eq.~\eqref{defcorrection} via an ancilla of dimension $Q+1$ and the unitary ${\rm select}(V)$, following \\ an analogous procedure to 1 (a)-(c). This yields the state in Eq.~\eqref{aftercor}.
\item Apply a single step of {\obv} on steps 1 and 2 above.
\end{enumerate}
\end{algorithm}

\section{Hamiltonian simulation with a corrected Taylor series}\label{sec:corrected_Taylor}
In this section we present and prove the main result of this work.
\begin{theorem}
\label{thm:corthm}
A Hamiltonian $H=\sum_{\ell=1}^ L \alpha_{\ell} H_{\ell}$, where $H_{\ell}$ are unitary matrices and $\alpha_{\ell}>0$ can be simulated for time $t$ within error $\epsilon>0$ with an overall complexity in terms of controlled-$H_\ell$ gates of
\begin{equation}
O\left(T \frac{\log T}{\log\log T} + \log(1/\epsilon) \right),
\end{equation}
where $T := \sum_{l=1}^L \alpha_{\ell} t$.
\end{theorem}

To prove this result we first need the following Lemma.
\begin{lemma}
\label{lemma1}
When simulating Hamiltonian evolution using Algorithm 1,
given that $\large\|\Hamseg-\appseg\|<\delta/r$, the correction operator $\appcor$
\begin{equation}\label{defcorrection1}
\appcor=\sum_{k=0}^Q a_k \mathcal{H}^k=\sum_j \eta_j \tilde{V}_j\, ,
\end{equation}
satisfies $\sum_j |\eta_j|\leq 1+ 2\delta + O(\delta^2/r)$.
\end{lemma}
Before proceeding with the proof of this Lemma, it is useful to define the functional $\sa$ which acts on a function $F$ of a matrix $X$, with $F(X)=\sum_n F_n X^n$, as
\begin{equation}
\sa(F)=\sum_n |F_n| A^n\, ,
\end{equation}
where $A$ is a positive number.
The motivation for this definition comes from the method of implementing an operator given by sums of unitaries, explained briefly in Sec.~\ref{background} and in more detail in \cite{BCCKS15}.
Expressing $\appseg$ as a function of $\mathcal{H}$, the state after implementing $\appseg(\mathcal{H})$ before OAA given in Eq.~\eqref{W_operator} can be written as
\begin{equation}
\frac{1}{\sa(\appseg)}\ket{0}\appseg(\mathcal{H})\ket{\psi}+\sqrt{1-\frac{1}{\sa(\appseg)^2}}\ket{\phi}\, ,
\end{equation}
where in this case $A=\sum_{\ell=1}^L \alpha_{\ell}$.
This is the value of $A$ that will be used throughout this paper.
The quantity $\sa(F)$ is thus related to the number of steps of OAA needed to amplify the component of $\ket{0}F(\mathcal{H})\ket{\psi}$ to $\approx 1$.
The definition here is slightly different from the one in Ref.~\cite{berry2016corrected}, because we include the factor $A$ which comes from the decomposition of $H$ in terms of unitary operators as in Eq.~\eqref{H_decomposition}.
It is simple to show that the functional $\sa$ obeys the following properties
for functions $F$ and $G$ and scalars $\alpha,\beta \in \mathbb{C}$,
\begin{align}
\sa(\alpha F +\beta G) &\le |\alpha|\sa(F) + |\beta|\sa(G)\, , \label{props1} \\
\sa(FG) &\le \sa(F)\sa(G)\, . \label{props2}
\end{align}
These are the same properties as in Ref.~\cite{berry2016corrected}.
The proofs are identical, so will not be given here.
These properties will be useful in what follows.

\begin{proof}[Proof of Lemma \ref{lemma1}] This proof is similar to that for Lemma 2 of Ref.~\cite{berry2016corrected}. 
The operator implemented by the  algorithm before the correction is applied can be written as 
\begin{equation}
\oaa^r={\appseg}^r{\left(\frac{3}{2}\openone -\frac{1}{2}\appseg^{\dagger}\appseg\right)}^r .
\end{equation}
Let us define a \emph{perfect} correction operator 
such that it gives exactly the desired operation; that is
\begin{equation}\label{perfect_correction}
\corr \oaa^r = U\, .
\end{equation}
We can write $\corr$ as 
\begin{equation}
\corr = U \appseg^{-r}{\left(\frac{3}{2}\openone-\frac{1}{2}\appseg^{\dagger}\appseg\right)}^{-r}.
\end{equation}
If we now define 
\begin{equation}\label{eq:delta}
\Delta := \Hamseg-\appseg \, ,
\end{equation}
then in the same way as in Ref.~\cite{berry2016corrected} we can express $\corr$ as 
 \begin{align}\label{corofW}
\corr &=\left(\mathbb{I}-W\right)^{-r}\nn
&=\sum_{k=0}^{\infty} {r+k-1 \choose k} W^k\, ,
\end{align}
with 
\begin{align}
W &=\frac{1}{2} \left[\appseg^{\dagger}\Delta-\Delta^{\dagger} \appseg+\Delta^{\dagger}\Delta+ (\appseg^{\dagger})^2\Delta^2+(\Delta^{\dagger})^2\Delta^2+2\appseg^{\dagger}\Delta^{\dagger}\Delta^2+\appseg^{\dagger} \appseg\Delta\Delta^{\dagger} + \appseg\Delta(\Delta^{\dagger})^2 \right]. \label{defW}
\end{align}

Using the properties of the functional $\sa$ given in Eqs.~\eqref{props1} and \eqref{props2},
and regarding $\corr$, $\Delta$, $W$ and so forth as functions of $\mathcal{H}$, we obtain
\begin{equation}
\sa(\corr)\leq \sum_{k=0}^{\infty} {r+k-1 \choose k} \sa(W)^k = [1-\sa(W)]^r \,.
\end{equation}
Furthermore, we have that 
\begin{equation}
\sa(\Delta)=\sa(\appseg-\Hamseg)=\sum_{k=K+1}^\infty \frac{(T/r)^k}{k!}\leq \delta/r \, ,
\end{equation}
and $\sa(\appseg) \le 2$. Hence, using Eq.~\eqref{defW} we obtain
\begin{equation}
\sa(W)\leq 2\sa(\Delta) + 9[\sa(\Delta)]^2+6[\sa(\Delta)]^3+[\sa(\Delta)]^4\, .
\end{equation}
This implies that $\sa(W)\le 2\delta/r + O(\delta^2/r^2)$, and hence 
\begin{align}
\sa(\corr) &= [1-\sa(W)]^{-r}\nn
&\leq [1-2\delta/r+O(\delta^2/r^2)]^{-r}\nn
&=1+2\delta+\mathcal{O}(\delta^2/r) \, .
\end{align}
The actual correction operator implemented, $\appcor$, satisfies $\sa(\appcor)<\sa(\corr)$.
The value of $\sa(\appcor)$ corresponds to the sum of the absolute values of $\eta_j$.
Hence we obtain $\sum_j|\eta_j|\le 1+2\delta+O(\delta^2/r)$, as required. 
\end{proof}

Next, we prove a Lemma regarding the value of $Q$ needed in order for the correction operator to be sufficiently accurate.
\begin{lemma}\label{lemma2}
When simulating a Hamiltonian using Algorithm 1, there exists a truncation
$Q=O(T+\log(1/\epsilon))$
such that the correction operator $\appcor$ yields final error no greater than $\epsilon$.
\end{lemma}

\begin{proof} From Eq.~\eqref{perfect_correction}, we obtain 
\begin{equation}
\left\|\appcor \oaa^r-U\right\|=\left\|(\corr-\appcor)\oaa^r\right\|\, ,
\end{equation}
so we need to choose a $Q$ such that 
$\|(\corr-\appcor)\oaa^r\|<\epsilon$.
We can bound this expression as
\begin{equation}
\left\|(\corr-\appcor)\oaa^r\right\|\leq\left\|\corr-\appcor\right\|\cdot \left\|\oaa^r\right\|\leq \left\|\corr-\appcor\right\|\, ,
\end{equation}
since $\left\|\oaa \right\|\leq 1$ so $\left\|\oaa^r\right\|\leq 1$.
In order to bound $\|\appcor-\corr\|$ it is useful to define the operator 
\begin{equation}\label{defV}
\vd :=\Hamseg^{\dagger}\Delta\, ,
\end{equation}
such that we can write $W$ as 
\begin{equation}\label{WofV}
W=\frac{\vd}{2}-\frac{\vd^{\dagger}}{2}+\vd^{\dagger}\vd+\frac{\vd^2}{2}-\frac{\vd^{\dagger}\vd^2}{2}\, .
\end{equation}
We can bound $\|\appcor-\corr\|$ using the following trick (also used in \cite{berry2016corrected}) valid for a constant $x\geq 1$:
\begin{align}\label{eq:findif}
\left\|\appcor-\corr\right\|&=\left\|\sum_{k=Q+1}^{\infty}a_k \mathcal{H}^k\right\| \nn
&\leq\sum_{k=Q+1}^{\infty}|a_k|~\| \mathcal{H}\|^k\nn
&\leq\sum_{k=Q+1}^{\infty}|a_k|~A^k\nn
&\leq\frac{1}{x^{Q+1}}\sum_{k=Q+1}^{\infty}|a_k|~(A x)^k\nn
&\leq\frac{1}{x^{Q+1}}\corrp(x)\, ,
\end{align}
where we used the fact that $\|H\|\leq \sum_{\ell=1}^L\alpha_\ell=A$ and defined 
\begin{equation}
\corrp{(x)}:=\sum_{k=0}^{\infty}|a_k|~(A x)^k\, .
\end{equation}
Also, let us define the coefficients $b_k$ and $c_k$ as the Taylor coefficients of $\vd$ and $W$, respectively, so that
\begin{align}
\vd=\sum_{k>K}^{\infty}b_k \mathcal{H}^k\, ,\qquad
W=\sum_{k>K}^{\infty}c_k \mathcal{H}^k\, .
\end{align}
Then we define the series with positive coefficients
\begin{align}
\vd^+(x):=\sum_{k>K}^{\infty}|b_k| (A x)^k\, ,\qquad
W^+(x):=\sum_{k>K}^{\infty}|c_k| (A x)^k\, .
\end{align}
Using Eq.~\eqref{corofW} we have 
\begin{equation}\label{eq:corp}
\corrp{(x)}\leq \left[\mathbb{I}-W^+(x)\right]^{-r} .
\end{equation}
From Eq.~\eqref{WofV} we can bound $W^+(x)$ as
\begin{equation}
W^+(x)\leq \vd^+(x)+\frac{3}{2}[\vd^+(x))]^2+\frac{1}{2}[\vd^+(x)]^3\, .
\end{equation}
At this point, we need to upper bound the coefficients $b_k$ of $\vd$. We have that  
\begin{align}
\vd&=\Hamseg^{\dagger}(\Hamseg-\appseg)\nn
&=\sum_{k=0}^{\infty}\sum_{k'>K}^{\infty} \frac{(-t/r)^k}{k!}\frac{(t/r)^{k'}}{k'!}\mathcal{H}^{k+k'}\nn
&=\sum_{m>K}^{\infty}\frac{(t/r)^m}{m!} \mathcal{H}^{m} \sum_{k>K}^m(-1)^{m-k}{m\choose k}\, . 
\end{align}
Hence, we obtain for $k>K$,
\begin{align}
|b_k|&=\frac{(t/r)^k}{k!} \sum_{j>K}^k(-1)^{k-j}{k\choose j}\nn
&\leq \frac{(t/r)^k}{k!} \sum_{j>K}^k {k\choose j}\nn
&\leq \frac{(2 t/r)^k}{k!}\, .
\end{align}
For $k\le K$ we have $b_k=0$. Using this bound, we obtain
\begin{align}
\vd^+(x)&=\sum_{k>K}^{\infty}|b_k| (A x)^k\nn
&\leq \sum_{k> K}^{\infty}\frac{(2t A x/r)^ n}{n!}\nn
&\leq\sum_{n> K}^{\infty}\frac{(2Tx/r)^n}{n!}\nn
&\le \sum_{n> K}^{\infty}\frac{(x/2)^n}{n!}\, .
\end{align}
In the last line we have used the fact that we have chosen $r\ge 4T$.
We also restrict to $K \geq 2$.
Choosing $x=2$, we obtain $\vd^+(x)\lesssim 0.22$ and $W^+(x)\lesssim 0.29$.
Then Eq.~\eqref{eq:corp} implies that 
\begin{equation}
\corrp(x)\leq 2^r \, .
\end{equation}
Using Eq.~\eqref{eq:findif} we obtain the bound 
\begin{equation}
\left\|\appcor-\corr\right\|\leq 2^{r-Q-1}\, .
\end{equation}
Thus, to ensure that the error in the whole quantum simulation algorithm is less than $\epsilon$, we choose $Q$ such that 
\begin{equation}
2^{r-Q-1}\leq \epsilon\,  .
\end{equation}
Because $r=O(T)$, this inequality can be achieved with
\begin{equation}
Q=O(T+\log(1/\epsilon)) \, .
\end{equation}
\end{proof}

Using the results of Lemma~\ref{lemma1} and Lemma~\ref{lemma2} we can finally prove the main result of our paper.\\

\begin{proof}[Proof of Theorem \ref{thm:corthm}.]
The second part of the algorithm requires the implementation of a correction operator 
\begin{equation}
\appcor=\sum_{k=0}^Q a_k \mathcal{H}^k.
\end{equation}
Since $H$ is given by a sum of unitary matrices $H_{\ell}$, we can write $\appcor$ as a sum of unitaries $\sum_j \eta_j \tilde{V}_j$, where each $\tilde{V}_j$ is given by a product of some of the unitary operators $H_{\ell}$.
The coefficients $a_k$ can be calculated from Eq.~\eqref{corofW} since the Taylor expansion of $W$ can be calculated from Eq.~\eqref{defW}.
The implementation of $\appcor$ is analogous to that of $\appseg$ from Eq.~\eqref{utilde}, and so it requires $O(Q)$ operations controlled on ancilla qubits \cite{BCCKS15} and a number of steps of OAA which depends on the quantity
\begin{equation}
\sa(\appcor)=\sum_j |\eta_j|\, .
\end{equation} 
In Lemma~\ref{lemma1}, we have shown that $\sa(\appcor)$ is bounded by $1+2\delta+O(\delta^2/r)$.
Therefore, by choosing $\delta$ to be slightly less than $1/2$, we find that $\sa(\appcor)\le 2$, and therefore the OAA can be achieved in a single step.
Moreover, in Lemma~\ref{lemma2} we have shown that the correction can yield error $<\epsilon$ with $Q=O(T+\log(1/\epsilon))$.
The complexity of the correction is therefore $O(T+\log(1/\epsilon))$.

The first part of the quantum simulation algorithm consists of implementing $\oaa^r$, which gives an approximation of the operator $U$ up to accuracy $\delta$. Using the results of \cite{BCCKS15}, the complexity is
\begin{equation}
O\left(T \frac{\log{(T/\delta)}}{\log\log(T/\delta)}\right) .
\end{equation}
Above we have found that we can take $\delta\approx 1/2$.
Therefore the complexity of the full algorithm is 
\begin{equation}
O\left(T \frac{\log{(T)}}{\log\log(T)}+\log(1/\epsilon)\right) .
\end{equation} \end{proof}
 
\section{Conclusions}\label{conclusions}
We have improved on the complexity of the truncated Taylor series method for Hamiltonian simulation by adding a new step to the procedure which involves the application of a correction operator. We have shown that this operator can be written as a sum of unitaries and that it can be implemented via Oblivious Amplitude Amplification. In general, the ideas presented in this paper and in \cite{berry2016corrected}, as well as the proof techniques used, are versatile enough to be applied to any implementation of sums of unitaries in a quantum computer. Furthermore, the truncated Taylor series method is applicable to the simulation of many Hamiltonians of interest, including quantum chemistry problems \cite{chem1, chem2}.
For this reason, the improvement to the complexity of the method presented in this work could be significant, particularly in the early stage of quantum computers  where the number of qubits and quantum gates available is highly limited. 

Further improvement could, in principle, be possible by considering multiple rounds of correction as discussed in \cite{berry2016corrected}. This would slightly improve the dependence of the complexity on the simulation time, but the proofs become much more intricate and so we leave that for future exploration.
   
\acknowledgements
We acknowledge support from IARPA contract number D15PC00242.
DWB is funded by an Australian Research Council Future Fellowship (FT100100761) and a Discovery Project (DP160102426).
LN thanks the support from Funda\c{c}\~{a}o para a Ci\^{e}ncia e a Tecnologia (Portugal), namely through programmes PTDC/POPH/POCH and projects UID/EEA/50008/2013, IT/QuSim, IT/QuNet, ProQuNet, partially funded by EU FEDER, and from the EU FP7 project PAPETS (GA 323901). Furthermore, LN acknowledges the support from the DP-PMI and FCT (Portugal) through scholarship SFRH/BD/52241/2013.


\end{document}